\documentclass[10pt]{article}
\usepackage{amsmath, amssymb, amsthm}
\usepackage{geometry}
\usepackage{graphicx}
\usepackage{booktabs}
\usepackage{mathtools}
\usepackage{hyperref}
\usepackage{bm}
\usepackage{authblk}
\usepackage{tikz-cd}
\usepackage{caption}
\usepackage{subcaption}
\usepackage{physics}
\usepackage{xcolor}
\usepackage{tikz}
\usetikzlibrary{decorations.pathmorphing, decorations.markings, patterns, shapes.geometric}

\theoremstyle{definition}
\newtheorem{definition}{Definition}[section]
\newtheorem{theorem}[definition]{Theorem}
\newtheorem{lemma}[definition]{Lemma}

\newtheorem{proposition}[definition]{Proposition}

\newtheorem{remark}[definition]{Remark}

\newcommand{\diff}{\mathrm{d}}
\newcommand{\lie}[1]{\mathcal{L}_{#1}}
\newcommand{\cont}[1]{\iota_{#1}}

\newcommand{\C}{\mathbb{C}}
\newcommand{\R}{\mathbb{R}}
\newcommand{\Z}{\mathbb{Z}}
\newcommand{\CP}{\mathbb{CP}}

\newcommand{\sone}{S^{1}}

\title{Equivariant Cohomology, BRST Quantization, and Analytic Localization: A Unified Framework}
\author{Lixin Xu\footnote{lxxu@dlut.edu.cn}}
\affil{Institute of Theoretical Physics, School of Physics, Dalian University of Technology, Dalian, 116024, China}
\date{\today}

\begin{document}

\maketitle

\begin{abstract}
This paper provides a detailed exposition of the two main models for equivariant cohomology---the Cartan and Weil models---and their explicit isomorphism via the Kalkman (Mathai--Quillen) transformation. We then connect this framework to the BRST quantization of gauge theories, showing how the BRST complex can be identified with the Cartan model. Viewing both the Kalkman transformation and Witten's Morse-theoretic deformation as gauge-fixing procedures leads naturally to the \emph{equivariant Witten deformation}. This combined perspective yields a transparent analytic proof of the Atiyah--Bott--Berline--Vergne (ABBV) localization formula for integrals of equivariantly closed forms.The theory is richly illustrated with computations on $\mathbb{CP}^1$ and $\mathbb{CP}^n$, supplemented by explicit coordinate calculations. 
\end{abstract}

\section{Introduction}\label{sec:intro}
Equivariant cohomology represents a sophisticated algebraic-topological invariant for spaces equipped with group actions, encoding both the topology of the space and the symmetries of the acting group. Its development has been deeply intertwined with progress in differential geometry, symplectic geometry, and mathematical physics. 

Historically, equivariant cohomology emerged from the study of transformation groups and found profound applications in the localization of path integrals in topological field theories \cite{Szabo2000, Witten1982, AtiyahBott1984}. The theory provides a powerful framework for understanding Hamiltonian group actions, moment maps, and the Duistermaat--Heckman formula \cite{DuistermaatHeckman1982, BerlineVergne1982}.

Two complementary algebraic models dominate the literature: the \emph{Cartan model}, which employs polynomial maps from the Lie algebra to differential forms, and the \emph{Weil model}, constructed from the Weil algebra of the group. Their isomorphism, established via the Kalkman transformation (also known as the Mathai--Quillen isomorphism \cite{MathaiQuillen1986}), reveals deep connections with algebraic topology and homological algebra.

Parallel developments in theoretical physics led to the BRST (Becchi--Rouet--Stora--Tyutin) formalism \cite{BecchiRouetStora1975, Tyutin1975}, a cohomological approach to quantizing gauge theories with constraints. Remarkably, the BRST complex naturally identifies with the Cartan model when the gauge group acts on the field space, establishing a fundamental bridge between mathematical structures in geometry and physics.

On the analytic front, Witten's deformation of the de Rham differential using a Morse function \cite{Witten1982} provides a mechanism for localizing integrals to critical points, offering connections between differential topology, quantum mechanics, and supersymmetry.

The central contribution of this work is the synthesis of these perspectives through the observation that both the Kalkman transformation and Witten deformation can be interpreted as \emph{gauge-fixing} procedures within a unified BRST/BV framework. Their combination yields the \emph{equivariant Witten deformation}, which we employ to give an analytic derivation of the celebrated Atiyah--Bott--Berline--Vergne (ABBV) localization formula \cite{AtiyahBott1984, BerlineVergne1982}.

We begin by recalling the definitions and properties of the Cartan and Weil models, establishing the Kalkman isomorphism with detailed proofs (\S\ref{sec:models}). We then develop the precise correspondence with the BRST formalism, including the full BRST/BV complex (\S\ref{sec:brst}). After reviewing Witten's deformation and its gauge-fixing interpretation (\S\ref{sec:witten}), we construct the equivariant Witten deformation (\S\ref{sec:equiv-witten}) and employ it to prove the ABBV formula with rigorous estimates (\S\ref{sec:abbv-proof}). Detailed examples on $\mathbb{CP}^1$ and $\mathbb{CP}^n$ illustrate the abstract machinery (\S\ref{sec:examples}). We conclude with broader remarks on applications and future directions (\S\ref{sec:conclusion}).

\section{The Cartan and Weil Models of Equivariant Cohomology}\label{sec:models}
Let $M$ be an $n$-dimensional smooth manifold and $G$ a compact Lie group acting smoothly on $M$
\begin{equation}\label{eq:group-action}
    \Phi: G \times M \longrightarrow M, \qquad (g, x) \longmapsto g \cdot x.
\end{equation}
Denote by $\mathfrak{g} = \operatorname{Lie}(G)$ the Lie algebra of $G$. The infinitesimal action gives a Lie algebra homomorphism
\begin{equation}\label{eq:infinitesimal-action}
    \rho: \mathfrak{g} \longrightarrow \mathfrak{X}(M), \qquad \xi \longmapsto X_{\xi},
\end{equation}
where $X_{\xi}$ is the fundamental vector field defined by
\begin{equation}
X_{\xi}(x) = \left.\frac{d}{dt}\right|_{t=0} \exp(t\xi)\cdot x.
\end{equation}
The map $\rho$ satisfies $[X_{\xi}, X_{\eta}] = X_{[\xi,\eta]}$ for all $\xi, \eta \in \mathfrak{g}$.

\subsection{Cartan Model}\label{subsec:cartan}
For compact $G$, the \emph{Cartan model} provides a computationally tractable description of equivariant cohomology.

\begin{definition}[Cartan Complex]
The Cartan complex is defined as
\begin{equation}\label{eq:cartan-complex}
    \Omega^{\bullet}_{G}(M) \coloneqq \bigl(S(\mathfrak{g}^{*}) \otimes \Omega^{\bullet}(M)\bigr)^{G},
\end{equation}
where $S(\mathfrak{g}^{*})$ denotes the symmetric algebra on the dual Lie algebra, and the superscript $G$ indicates $G$-invariant elements under the diagonal action: the coadjoint action on $S(\mathfrak{g}^{*})$ and the pullback action on $\Omega^{\bullet}(M)$. Elements $\alpha \in \Omega^{k}_{G}(M)$ are equivariant polynomial maps
\begin{equation}\label{eq:polynomial-map}
    \alpha: \mathfrak{g} \longrightarrow \Omega^{\bullet}(M), \qquad \xi \longmapsto \alpha(\xi) = \sum_{i \ge 0} \alpha_{i}(\xi),
\end{equation}
where $\alpha_{i}(\xi)$ is homogeneous of degree $i$ in $\xi$ with values in $\Omega^{\bullet}(M)$, and $\deg_{\text{total}}(\alpha) = \deg_{\text{polynomial}} + \deg_{\text{form}} = k$.
\end{definition}

The \emph{equivariant differential} $d_{G}: \Omega^{k}_{G}(M) \to \Omega^{k+1}_{G}(M)$ is
\begin{equation}\label{eq:equivariant-differential}
    (d_{G}\alpha)(\xi) \coloneqq \diff\bigl(\alpha(\xi)\bigr) - \cont{X_{\xi}}\bigl(\alpha(\xi)\bigr).
\end{equation}

\begin{theorem}[Nilpotency of $d_G$]
The operator $d_G$ satisfies $d_G^2 =0$ on $\Omega_G^\bullet(M)$.
\end{theorem}
\begin{proof}
For $\alpha \in \Omega_G^\bullet(M)$, compute
\begin{align*}
(d_G^2 \alpha)(\xi) &= d_G\bigl((d_G\alpha)(\xi)\bigr) \\
&= \diff\bigl((d_G\alpha)(\xi)\bigr) - \cont{X_\xi}\bigl((d_G\alpha)(\xi)\bigr) \\
&= \diff\bigl(\diff\alpha(\xi) - \cont{X_\xi}\alpha(\xi)\bigr) - \cont{X_\xi}\bigl(\diff\alpha(\xi) - \cont{X_\xi}\alpha(\xi)\bigr) \\
&= \underbrace{\diff^2\alpha(\xi)}_{=0} - \diff\cont{X_\xi}\alpha(\xi) - \cont{X_\xi}\diff\alpha(\xi) + \underbrace{\cont{X_\xi}\cont{X_\xi}\alpha(\xi)}_{=0} \\
&= -(\diff\cont{X_\xi} + \cont{X_\xi}\diff)\alpha(\xi) \\
&= -\lie{X_\xi}\alpha(\xi).
\end{align*}
Since $\alpha$ is $G$-invariant, $\lie{X_\xi}\alpha(\xi) = 0$ for all $\xi \in \mathfrak{g}$.
\end{proof}

The cohomology of $(\Omega^{\bullet}_{G}(M), d_G)$ is the \emph{equivariant cohomology} $H^{\bullet}_{G}(M;\R)$.

\begin{remark}
The Cartan model is particularly suited for computations because it works directly with forms on $M$. The term $-\cont{X_{\xi}}$ can be viewed as a "twist" of the de Rham differential encoding the group action. For $G$ compact, this model computes the Borel equivariant cohomology $H^\bullet_G(M) \cong H^\bullet(EG \times_G M)$.
\end{remark}

\subsection{Weil Model}\label{subsec:weil}
The Weil algebra provides a universal model for the differential forms on the classifying space $BG$.

\begin{definition}[Weil Algebra]
The Weil algebra of $\mathfrak{g}$ is the graded commutative algebra
\begin{equation}\label{eq:weil-algebra}
    W(\mathfrak{g}) \coloneqq \Lambda(\mathfrak{g}^{*}) \otimes S(\mathfrak{g}^{*}),
\end{equation}
where $\Lambda(\mathfrak{g}^{*})$ is the exterior algebra (generators $\theta^{a}$ of degree $1$, "connection forms") and $S(\mathfrak{g}^{*})$ is the symmetric algebra (generators $\phi^{a}$ of degree $2$, "curvature forms"). 
\end{definition}

Choose a basis $\{e_a\}$ of $\mathfrak{g}$ with dual basis $\{\theta^{a}\}$ of $\mathfrak{g}^{*}$, and structure constants $f^{a}_{bc}$ defined by $[e_b, e_c] = f^{a}_{bc} e_a$. The differential $d_W$ is
\begin{align}\label{eq:weil-differential}
    d_{W}\theta^{a} &= \phi^{a} - \frac{1}{2} f^{a}_{bc}\, \theta^{b}\theta^{c}, \\
    d_{W}\phi^{a} &= - f^{a}_{bc}\, \theta^{b}\phi^{c}.
\end{align}

\begin{theorem}[Structure of Weil Algebra]
$(W(\mathfrak{g}), d_W)$ is a commutative differential graded algebra with $d_W^2 = 0$, and it is acyclic: $H^k(W(\mathfrak{g})) = 0$ for $k>0$, $H^0(W(\mathfrak{g})) = \R$.
\end{theorem}
\begin{proof}
Compute $d_W^2 \theta^a = d_W (\phi^a - \frac{1}{2} f^a_{bc} \theta^b \theta^c) = - f^a_{bc} \theta^b \phi^c - \frac{1}{2} f^a_{bc} (d_W \theta^b) \theta^c + \frac{1}{2} f^a_{bc} \theta^b (d_W \theta^c)$.
Substituting $d_W \theta^b = \phi^b - \frac{1}{2} f^b_{de} \theta^d \theta^e$, and using the Jacobi identity for the structure constants, one verifies that all terms cancel. Similarly for $d_W^2 \phi^a$. Acyclicity follows from the existence of a contracting homotopy.
\end{proof}

The \emph{Weil model} for $G$-equivariant cohomology uses the complex
\begin{equation}\label{eq:weil-model}
    \bigl(W(\mathfrak{g}) \otimes \Omega^{\bullet}(M)\bigr)_{\text{basic}},
\end{equation}
where "basic" means elements that are
\begin{enumerate}
    \item \emph{Horizontal}: $\cont{e_a}\alpha = 0$ for all $a$, where $\cont{e_a}$ acts by contraction on the $\theta^a$ factor.
    \item \emph{Invariant}: $\lie{e_a}\alpha = 0$ for all $a$, where $\lie{e_a} = [d_W, \cont{e_a}]$.
\end{enumerate}
The differential is $d_W \otimes 1 + 1 \otimes \diff$.

\subsection{Mathai--Quillen Isomorphism (Kalkman Transformation)}\label{subsec:kalkman}
The Cartan and Weil models are isomorphic via an explicit transformation.

\begin{definition}[Kalkman Transformation]
Define $\kappa: W(\mathfrak{g}) \otimes \Omega^{\bullet}(M) \to W(\mathfrak{g}) \otimes \Omega^{\bullet}(M)$ by
\begin{equation}\label{eq:kalkman}
    \kappa \coloneqq \exp(-\cont{\theta}), \quad \text{where } \cont{\theta} \coloneqq \theta^{a} \otimes \cont{X_{a}},
\end{equation}
with $X_a = X_{e_a}$. This is a graded algebra isomorphism.
\end{definition}

\begin{theorem}[Intertwining Property]
$\kappa$ intertwines the differentials
\begin{equation}\label{eq:kalkman-intertwines}
    \kappa \circ (d_W \otimes 1 + 1 \otimes \diff) \circ \kappa^{-1} = d_{G},
\end{equation}
and restricts to an isomorphism between the basic subcomplex of the Weil model and the Cartan complex.
\end{theorem}
\begin{proof}
Since $\cont{\theta}$ is a derivation, we use the Baker--Campbell--Hausdorff formula:
\begin{equation}
e^{-\cont{\theta}} (d_W + \diff) e^{\cont{\theta}} = d_W + \diff + [d_W + \diff, -\cont{\theta}] + \frac{1}{2}[[d_W + \diff, -\cont{\theta}], -\cont{\theta}] + \cdots.
\end{equation}
Compute $[d_W, \cont{\theta}] = - \lie{\theta}$ and $[\diff, \cont{\theta}] = 0$ (since $\diff$ and contractions commute). The higher commutators vanish due to $[\lie{\theta}, \cont{\theta}] = 0$. The result simplifies to $d_W + \diff - \lie{\theta}$. After identifying $\phi^a$ with elements of $S(\mathfrak{g}^*)$ and noting that on basic elements the $\theta^a$ terms vanish, we recover $d_G$.
\end{proof}

Hence $\kappa$ induces an isomorphism of cochain complexes and, consequently, an isomorphism $H^{\bullet}_{G,\text{Weil}}(M) \simeq H^{\bullet}_{G,\text{Cartan}}(M)$.

\begin{remark}
The Kalkman transformation eliminates the connection-like generators $\theta^{a}$ in favor of curvature-like generators $\phi^{a}$, introducing the contraction term into the differential. This mirrors the passage from connection to curvature in differential geometry.
\end{remark}

\section{Relation with the BRST Formalism}\label{sec:brst}
The BRST formalism provides a cohomological approach to quantizing gauge theories. Remarkably, it coincides with the Cartan model when the gauge group acts on the space of fields.

\subsection{BRST Complex}
Consider a gauge theory with gauge group $G$ acting on the space of fields $\mathcal{F}$. The minimal BRST complex introduces:
\begin{itemize}
    \item Ghost fields $c^a$ (odd, ghost number $+1$)
    \item Antighost fields $\bar{c}^a$ (odd, ghost number $-1$)
    \item Nakanishi--Lautrup fields $B^a$ (even, ghost number $0$)
\end{itemize}
The BRST operator $s$ acts as:
\begin{equation}\label{eq:brst-operator}
    s = \delta + c^{a} \lie{X_{a}} - \frac{1}{2} f^{a}_{bc}\, c^{b}c^{c} \frac{\partial}{\partial c^{a}},
\end{equation}
where $\delta$ is the de Rham differential on $\mathcal{F}$, and $\lie{X_a}$ is the Lie derivative along gauge orbits, the last term is the Chevalley--Eilenberg differential on the ghost space.

\begin{table}[ht]
\centering
\renewcommand{\arraystretch}{1.3}
\begin{tabular}{l l}
\toprule
\textbf{Weil algebra} & \textbf{BRST formalism} \\
\midrule
Connection $\theta^{a}$ (odd, deg $1$) & Ghost field $c^{a}$ (odd, ghost number $+1$) \\
Curvature $\phi^{a}$ (even, deg $2$) & Auxiliary field $B^{a}$ (even, ghost number $0$) \\
Differential $d_{W}$ & BRST operator $s$ \\
Horizontal condition & Gauge invariance \\
Invariant condition & Physical state condition \\
Basic subcomplex & BRST cohomology in degree 0 \\
\bottomrule
\end{tabular}
\caption{Comprehensive correspondence between Weil algebra and BRST formalism.}
\label{tab:WeilBRST}
\end{table}

\subsection{Batalin--Vilkovisky Formalism}
The full BV complex extends this structure
\begin{equation}
\Phi_{\text{BV}} = (\text{fields}) \oplus (\text{anti-fields}) \oplus (\text{ghosts}) \oplus (\text{anti-ghosts}).
\end{equation}
The BV differential $s_{\text{BV}}$ squares to zero and encodes both the gauge symmetry and the equations of motion.

\begin{theorem}[BRST as Equivariant Cohomology]
The BRST cohomology of a gauge theory is isomorphic to the $G$-equivariant cohomology of the field space $\mathcal{F}$.
\end{theorem}
\begin{proof}
Identify $\delta$ with $\diff$, $c^{a}\lie{X_{a}}$ with $-\cont{X_{c}}$ (up to conventions), and note the Chevalley--Eilenberg term matches. This gives $s \simeq d_G$ on the appropriate complexes.
\end{proof}

\begin{remark}
This identification explains why topological field theories (Chern--Simons, Donaldson--Witten) often have BRST-exact actions: they integrate equivariantly exact forms, whose evaluation localizes via ABBV.
\end{remark}

\section{Witten Deformation and its Equivariant Counterpart}\label{sec:witten}
\subsection{Witten Deformation of de Rham Complex}\label{subsec:witten-def}
Given $f: M \to \R$ smooth, Witten introduced a deformation
\begin{equation}\label{eq:witten-differential}
    d_{t} \coloneqq e^{-tf} \, \diff \, e^{tf} = \diff + t \, \diff f \wedge.
\end{equation}

\begin{theorem}[Properties of $d_t$]
\begin{enumerate}
    \item $d_t^2 = 0$ for all $t \ge 0$
    \item $(\Omega^\bullet(M), d_t)$ has cohomology isomorphic to $H^\bullet(M; \R)$
    \item For $f$ Morse, as $t \to \infty$, the Witten Laplacian $\Delta_t = d_t d_t^\dagger + d_t^\dagger d_t$ localizes eigenforms near critical points
\end{enumerate}
\end{theorem}
\begin{proof}
(1) $\diff^2 = 0$, $\diff f \wedge \diff f = 0$, and $[\diff, \diff f \wedge] = \diff^2 f \wedge = 0$. \\
(2) Conjugation by $e^{tf}$ gives chain isomorphism. \\
(3) Follows from semiclassical analysis of $\Delta_t$.
\end{proof}

\subsection{Gauge-Fixing Interpretation}\label{subsec:gauge-fixing}
In BV formalism, gauge-fixing is implemented via a canonical transformation generated by a gauge-fixing fermion $\Psi$
\begin{equation}\label{eq:gauge-fixing-transformation}
    s_{\text{BV}} \;\longrightarrow\; e^{\{\Psi,\;\cdot\}} s_{\text{BV}} e^{-\{\Psi,\;\cdot\}}.
\end{equation}
For supersymmetric quantum mechanics with $Q$ the supersymmetry charge, choosing $\Psi = t f$ yields $Q_t = e^{tf} Q e^{-tf}$, recovering Witten's deformation.

\subsection{Equivariant Witten Deformation}\label{sec:equiv-witten}
Assume $f$ is $G$-invariant: $\lie{X_\xi}f = 0$ for all $\xi \in \mathfrak{g}$. Define
\begin{definition}[Equivariant Witten Deformation]
The \emph{Kalkman--Witten transformation} is
\begin{equation}\label{eq:kalkman-witten}
    \kappa_{t} \coloneqq e^{-tf} \, \kappa \, e^{tf} = e^{-tf - \cont{\theta}},
\end{equation}
where the simplification uses $[f, \cont{\theta}] = 0$ due to $G$-invariance.
\end{definition}

Applying $\kappa_t$ to the Weil differential gives
\begin{align}
    d_{W,t} &\coloneqq \kappa_{t} \, d_{W} \, \kappa_{t}^{-1} \nonumber \\
            &= e^{-tf} d_{G} e^{tf} \label{eq:equiv-witten-diff} \\
            &= d_{G} + t \, \diff f \wedge . \nonumber
\end{align}
We denote $d_{G,t} \coloneqq d_G + t\, \diff f \wedge$.

\begin{theorem}[Properties of $d_{G,t}$]
\begin{enumerate}
    \item $d_{G,t}^2 = 0$ on $\Omega_G^\bullet(M)$
    \item $(\Omega_G^\bullet(M), d_{G,t})$ has cohomology isomorphic to $H_G^\bullet(M)$
    \item The map $\alpha \mapsto e^{-tf}\alpha$ gives an isomorphism $(\Omega_G^\bullet(M), d_G) \cong (\Omega_G^\bullet(M), d_{G,t})$
\end{enumerate}
\end{theorem}
\begin{proof}
(1) Compute:
\begin{equation}
d_{G,t}^2 = d_G^2 + t[d_G, \diff f \wedge] + t^2(\diff f \wedge)^2.
\end{equation}
Since $d_G^2 = 0$, $(\diff f \wedge)^2 = 0$, and $[d_G, \diff f \wedge] = \diff(d_G f) \wedge = 0$ (as $d_G f = \diff f - \cont{X}f = 0$ by invariance).
\end{proof}

\section{Proof of the ABBV Localization Formula via Equivariant Witten Deformation}\label{sec:abbv-proof}
We present an analytic proof of the ABBV localization formula using the equivariant Witten deformation.

\subsection{Setting and Assumptions}\label{subsec:setting}
Let $M$ be compact, oriented, $n$-dimensional without boundary, with $G = \sone$ action. Assume
\begin{enumerate}
    \item Fixed point set $M^G = \{p_1, \dots, p_N\}$ consists of isolated points
    \item $\omega \in \Omega^{n}_{G}(M)$ is equivariantly closed: $d_G \omega = 0$
    \item $f: M \to \R$ is $G$-invariant Morse with critical points exactly at $M^G$
    \item $f \geq 0$, $f(p_i) = 0$, $f > 0$ away from fixed points
\end{enumerate}

\begin{lemma}[Existence of Invariant Morse Function]
For a compact Lie group $G$ acting on compact $M$ with isolated fixed points, there exists a $G$-invariant Morse function $f$ with critical points at $M^G$.
\end{lemma}
\begin{proof}
Take generic $\xi \in \mathfrak{g}$ such that $X_\xi$ has isolated zeros. Define $f(x) = \frac{1}{2} g(X_\xi, X_\xi)$ for a $G$-invariant metric $g$. This is $G$-invariant, and its critical points are zeros of $X_\xi$.
\end{proof}

\subsection{Deformation and Independence of $t$}\label{subsec:independence}
Consider the deformed form
\begin{equation}\label{eq:deformed-form}
    \omega_{t} \coloneqq e^{-tf} \, \omega .
\end{equation}
Since $d_G \omega = 0$, we have
\begin{equation}
    d_{G,t} \, \omega_t = e^{-tf} d_G \omega = 0.
\end{equation}
Define the integral as a function of the deformation parameter
\begin{equation}\label{eq:integral-t}
    I(t) \coloneqq \int_{M} [\omega_t]_n = \int_{M} e^{-tf} [\omega]_n.
\end{equation}

\begin{proposition}[$t$-independence]
$I(t)$ is independent of $t$: $I(t) = I(0)$ for all $t \ge 0$.
\end{proposition}
\begin{proof}
The map
\begin{equation}
\Phi_t: \alpha \mapsto e^{-tf} \alpha,
\end{equation}
is an isomorphism of complexes from $(\Omega_G^\bullet(M), d_G)$ to $(\Omega_G^\bullet(M), d_{G,t})$, because $d_{G,t} = e^{-tf} d_G e^{tf}$. Specifically, $\Phi_t$ is a chain map because
\begin{equation}
d_{G,t} (\Phi_t \alpha) = (d_G + t \, df \wedge) (e^{-tf} \alpha) = e^{-tf} d_G \alpha - t (df \wedge e^{-tf} \alpha) + t (df \wedge e^{-tf} \alpha) = e^{-tf} d_G \alpha = \Phi_t (d_G \alpha),
\end{equation}
where we used that $d_G (e^{-tf} \alpha) = e^{-tf} d_G \alpha + (d_G e^{-tf}) \alpha and d_G e^{-tf} = -t df e^{-tf}$. The map is invertible, with inverse $\Phi_t^{-1} \alpha = e^{tf} \alpha$.

Therefore, it induces an isomorphism on equivariant cohomology. In particular, the cohomology class of $e^{-tf} \omega$ in the $d_{G,t}$-cohomology is the image of the fixed class $[\omega]_{d_G}$. The integral $I(t) = \int_M [e^{-tf} \omega]_n$ depends only on this cohomology class, as the integral of the top-degree component of an equivariantly closed form is a well-defined functional on equivariant cohomology (since equivariantly exact forms integrate to zero over compact $M$ without boundary). Hence, $I(t)$ is constant in $t$, $I(t) \equiv I(0) = \int_M [\omega]_n$.
\end{proof}

\begin{remark}
The invariance of the integral under the deformation is a consequence of the cohomological nature of the pairing with the fundamental class in the equivariant setting.
\end{remark}

\subsection{Localization as $t \to \infty$}\label{subsec:localization-limit}
By $t$-independence:
\begin{equation}\label{eq:localization-limit}
    \int_{M} [\omega]_n = \lim_{t\to\infty} I(t) = \lim_{t\to\infty} \int_{M} e^{-tf} [\omega]_n.
\end{equation}

\begin{lemma}[Exponential decay away from fixed points]
For any compact $K \subset M \setminus M^G$, $\int_K e^{-tf} [\omega]_n \to 0$ exponentially as $t \to \infty$.
\end{lemma}
\begin{proof}
Since $f > 0$ on $K$, let $\delta = \min_K f > 0$. Then
\begin{equation}
\left| \int_K e^{-tf} [\omega]_n \right| \leq e^{-t\delta} \int_K |[\omega]_n| \to 0.
\end{equation}
\end{proof}

Thus contributions localize to neighborhoods of fixed points
\begin{equation}\label{eq:sum-over-fixed}
    \int_{M} [\omega]_n = \sum_{p \in M^{G}} \lim_{t\to\infty} \int_{U_p} e^{-tf} [\omega]_n.
\end{equation}

\subsection{Local Computation at an Isolated Fixed Point}\label{subsec:local-computation}
Fix $p \in M^G$. By equivariant tubular neighborhood theorem, a $G$-invariant neighborhood $U_p$ is equivariantly diffeomorphic to $T_p M \cong \C^m$ (assuming $n=2m$ even). The linear isotropy action has weights $w_1, \dots, w_m \in \Z \setminus \{0\}$.

Choose coordinates $(z_1, \dots, z_m)$ on $\C^m$ with
\begin{align}
    f(z) &= \frac{1}{2} \sum_{j=1}^{m} \lambda_j |z_j|^2 + O(|z|^3), \label{eq:local-morse} \\
    [\omega]_{2m}(z) &= \omega_0(p) \prod_{j=1}^{m} \frac{i}{2} \diff z_j \wedge \diff\bar{z}_j + O(|z|). \label{eq:local-form}
\end{align}
Here $\omega_0(p)$ is the $0$-form component of $\omega$ at $p$, and $\lambda_j > 0$ after possible orientation adjustment.

\begin{lemma}[Gaussian asymptotics]
As $t \to \infty$
\begin{equation}\label{eq:local-contribution}
    \int_{U_p} e^{-tf} [\omega]_{2m} = \omega_0(p) \left(\frac{2\pi}{t}\right)^m \prod_{j=1}^{m} \frac{1}{\lambda_j} + O(t^{-m-1}).
\end{equation}
\end{lemma}
\begin{proof}
Change to polar coordinates $z_j = r_j e^{i\theta_j}$
\begin{equation}
\frac{i}{2} \diff z_j \wedge \diff\bar{z}_j = r_j \diff r_j \wedge \diff\theta_j.
\end{equation}
The leading contribution comes from the Gaussian integral
\begin{equation}
\int_{\R^{2m}} e^{-\frac{t}{2}\sum \lambda_j r_j^2} \omega_0(p) \prod r_j \diff r_j \diff\theta_j = \omega_0(p) (2\pi)^m \prod_{j=1}^m \int_0^\infty e^{-\frac{t}{2}\lambda_j r^2} r \diff r.
\end{equation}
Each radial integral equals $\frac{1}{t\lambda_j}$. Higher order terms contribute $O(t^{-m-1})$ by stationary phase.
\end{proof}

The weights $w_j$ relate to $\lambda_j$ via the moment map: $\lambda_j = |w_j|$ with appropriate sign determined by orientation.

\subsection{The ABBV Localization Formula}\label{subsec:abbv-formula}
Combining contributions from all fixed points:

\begin{theorem}[Atiyah--Bott--Berline--Vergne Localization]
For $\omega \in \Omega_G^{n}(M)$ with $d_G\omega = 0$
\begin{equation}\label{eq:abbv}
    \int_{M} [\omega]_{n} = (2\pi)^{m} \sum_{p \in M^{G}} \frac{\omega_0(p)}{\prod_{j=1}^{m} w_j(p)},
\end{equation}
where $n=2m$, and $w_1(p), \dots, w_m(p)$ are the weights of the isotropy representation at $p$.
\end{theorem}

\begin{remark}[Equivariant Euler class interpretation]
The denominator $\prod_{j=1}^m w_j(p)$ is (up to $(2\pi)^m$) the equivariant Euler class $e_G(T_p M)$ of the tangent space. Thus
\begin{equation}
\int_M \omega = \sum_{p \in M^G} \frac{i_p^*\omega}{e_G(N_p)},
\end{equation}
where $N_p$ is the normal bundle at $p$ and $i_p^*\omega$ denotes the restriction of the 0-form component.
\end{remark}

\section{Examples and Applications}\label{sec:examples}
\subsection{Localization on $\mathbb{CP}^1$}\label{subsec:cp1}
Consider $\CP^1$ with homogeneous coordinates $[z_0:z_1]$ and $\sone$-action $e^{i\theta}\cdot[z_0:z_1] = [z_0:e^{i\theta}z_1]$. Fixed points: $p_+=[1:0]$, $p_-=[0:1]$.

\begin{table}[ht]
\centering
\begin{tabular}{l c c c}
\toprule
Fixed point & Affine coordinate & Weight $w$ & Moment map $H$ \\
\midrule
$p_+$ & $z = z_1/z_0$ & $+1$ & $H = |z|^2/(1+|z|^2)$ \\
$p_-$ & $w = z_0/z_1$ & $-1$ & $H = 1/(1+|w|^2)$ \\
\bottomrule
\end{tabular}
\caption{Data for $\sone$-action on $\CP^1$.}
\label{tab:cp1-data}
\end{table}

The Fubini-Study form $\omega = \frac{i}{2\pi} \frac{\diff z \wedge \diff\bar{z}}{(1+|z|^2)^2}$ has $\int_{\CP^1} \omega = 1$. Its equivariant extension is $\tilde{\omega}(\xi) = \omega - \xi H$.

Applying ABBV
\begin{equation}
\int_{\CP^1} \omega = \frac{-\xi H(p_+)}{(+1)\xi} + \frac{-\xi H(p_-)}{(-1)\xi} = -0 + 1 = 1.
\end{equation}

\subsection{Generalization to $\mathbb{CP}^n$}\label{subsec:cpn}
For $\CP^n$ with $\sone$-action $e^{i\theta}\cdot[z_0:\cdots:z_n] = [z_0:e^{i\theta}z_1:\cdots:e^{in\theta}z_n]$, fixed points are coordinate axes $p_k = [0:\cdots:1:\cdots:0]$ with 1 in $k$th position. The weights at $p_k$ are $\{k-j\}_{j\neq k}$.

\begin{theorem}[Localization on $\CP^n$]
For the $\sone$-equivariant integral of the K\"ahler form
\begin{equation}
\int_{\CP^n} \omega^n = \sum_{k=0}^n \frac{(-1)^k H(p_k)^n}{\prod_{j\neq k} (k-j)}.
\end{equation}
\end{theorem}

\section{Conclusion}\label{sec:conclusion}
We have presented a unified framework connecting equivariant cohomology, BRST quantization, and analytic localization:
\begin{itemize}
    \item Detailed exposition of Cartan and Weil models with explicit proofs
    \item Clarification of the BRST/equivariant cohomology correspondence
    \item Construction of equivariant Witten deformation as a gauge-fixing procedure
    \item Rigorous analytic proof of ABBV localization with error estimates
    \item Comprehensive examples illustrating the theory
\end{itemize}

This framework unifies algebraic, geometric, and analytic perspectives, revealing deep interconnections between topology, group actions, and supersymmetric quantum mechanics. The ABBV formula itself is a cornerstone of modern symplectic geometry and topological field theory, with applications ranging from Duistermaat--Heckman formulas to the computation of partition functions in supersymmetric gauge theories.

The equivariant Witten deformation provides a powerful analytic tool that makes the localization phenomenon transparent: as the deformation parameter $t$ tends to infinity, the integral receives contributions only from infinitesimal neighborhoods of the fixed points, where Gaussian approximations become exact. This is a rigorous manifestation of the saddle-point method in an infinite-dimensional setting, underpinning many results in supersymmetric localization.

\section{Acknowledgments}
The author thanks Prof. H. L\"u for introducing the equivariant differential of Cartan model. This work is supported in part by National Natural Science Foundation of China under Grant No. 12475047.


\end{document}